\documentclass{article}
\usepackage{amsfonts}
\usepackage{amssymb}
\usepackage{amsthm}
\usepackage{amsmath}

\usepackage{graphicx}

\usepackage{subfigure}
\usepackage{bm}             
\usepackage[mathscr]{eucal}
\usepackage{cancel}
\usepackage{psfrag}
\usepackage{wasysym}

\usepackage{oubraces}

\def\pmb#1{\setbox0=\hbox{$#1$}%
  \kern-.025em\copy0\kern-\wd0
  \kern.05em\copy0\kern-\wd0
  \kern-.025em\raise.0433em\box0}
\def\pmbs#1{\setbox0=\hbox{$\scriptstyle #1$}%
  \kern-.0175em\copy0\kern-\wd0
  \kern.035em\copy0\kern-\wd0
  \kern-.0175em\raise.0303em\box0}

\def\be{\begin{equation}}
\def\ee{\end{equation}}
\def\bea{\begin{eqnarray}}
\def\eea{\end{eqnarray}}
\def\lb{\label}

\def\Sig{\Sigma}
\def\Sigp{\Sigma_{+}}
\def\Sigm{\Sigma_{-}}
\def\Sigc{\Sigma_{\times}}

\def\Nm{N_{-}}
\def\Nc{N_{\times}}

\def\Udot{\dot{U}}
\newcommand{\cu}{{\mathcal U}}

\def\cK{{\cal K}}

\def\ca{{\cal A}}

\def\ptl{\partial}

\def\hsp5{\hspace{5mm}}
\newcommand{\sfrac}[2]{{\textstyle{#1\over#2}}}

\def\case#1/#2{\textstyle\frac{#1}{#2}}

\renewcommand{\u}{\mathsf{u}}
\newcommand{\mZ}{\mathrm{mod}\:2}
\newcommand{\z}{z}
\newcommand{\kh}{K}
\newcommand{\khl}{L}
\newcommand{\khp}{\bar{K}}
\newcommand{\khpl}{\bar{L}}

\newcommand{\Thi}{\mathcal{T}_{\mathsf{Hi}}}

\newcommand{\Tco}{\mathcal{T}_{\mathsf{Jo}}}

\newcommand{\era}{\overline{\mathrm{era}}}

\newcommand{\marginalnote}[1]{\marginpar{\begin{footnotesize}{\it #1}\end{footnotesize}}}

\theoremstyle{plain}
\newtheorem{theorem}{Theorem}[section]
\newtheorem{corollary}[theorem]{Corollary}

\newtheorem{lemma}[theorem]{Lemma}

\theoremstyle{remark}

\newtheorem*{remark}{Remark}

\begin{document}

\title{\sc Spike statistics}
\author{
\sc J.\ Mark Heinzle$^{1}$\thanks{Electronic address: {\tt
mark.heinzle@univie.ac.at}}\, and Claes Uggla$^{2}$\thanks{Electronic
address:{\tt claes.uggla@kau.se}}\\
$^{1}${\small\em University of Vienna, Faculty of Physics,}\\
{\small\em Gravitational Physics, Boltzmanngasse 5, 1090 Vienna, Austria}\\
$^{2}${\small\em Department of Physics, University of Karlstad,}\\
{\small\em S-65188 Karlstad, Sweden} }

\date{}
\maketitle
\begin{abstract}
In this paper we explore stochastical and statistical properties of so-called
recurring spike induced Kasner sequences. Such sequences arise in recurring
spike formation, which is needed together with the more familiar BKL scenario
to yield a complete description of generic spacelike singularities. In
particular we derive a probability distribution for recurring spike induced
Kasner sequences, complementing similar available BKL results, which makes
comparisons possible. As examples of applications, we derive results for
so-called large and small curvature phases and the Hubble-normalized Weyl
scalar.

\end{abstract}
\centerline{\bigskip\noindent PACS number(s):
04.20.Dw, 04.20.Ha, 04.20.-q, 05.45.-a, 98.80.Jk\hfill 
}\vfill

\section{Introduction}\label{sec:intro}

It is no understatement to say that the work of Belinski\v{\i}, Khalatnikov
and Lifshitz (BKL)~\cite{lk63,bkl70,bkl82} has set the stage for much of
subsequent investigations about the detailed nature of generic spacelike
singularities. The starting point for their analysis is their
\emph{`locality' conjecture}, which states that asymptotically toward a
generic spacelike singularity in \emph{inhomogeneous} cosmology the dynamics
is local, in the sense that each spatial point is assumed to evolve toward
the singularity individually and independently of its neighbors as a
\emph{spatially homogeneous} model~\cite{lk63,bkl82}. A second important
conjecture of BKL is that some sources, like perfect fluids with sufficiently
soft equations of state such as dust or radiation, lead to models that
generically are asymptotically `\emph{vacuum dominated}', i.e.,
asymptotically toward a generic spacelike singularity the spacetime geometry
is not influenced by the matter content, even though, e.g., the energy
density blow up~\cite{lk63,bkl70,bkl82}. This leads to a picture of
temporally oscillating Kasner states (i.e., vacuum Bianchi type I solutions,
which have flat spatial curvature) where two subsequent Kasner states are
determined by the vacuum Bianchi type II solutions, so-called single
curvature transitions in the nomenclature of~\cite{heietal09}.\footnote{The
term transition denotes a change between two different Kasner states; `single
curvature' refers to a single degree of freedom that describes the spatial
curvature, which is excited during the transition.}

Several recent papers~\cite{heietal09}--\cite{ugg13} have provided further
evidence for the BKL scenario, but it has also been found that there are
timelines whose evolution is different, exhibiting \emph{`non-local'
recurring spike formation}. Remarkably, BKL and recurring spike behavior are
intimately linked in a hierarchical manner: by hierarchies of subsets and by
a solution generating algorithm that links all building blocks for
understanding asymptotics at generic spacelike singularities to the bottom of
the hierarchy, which consists of the Kasner solutions.

In a recent paper~\cite{heietal12} it was shown that recurring spike
formation could be analytically described by means of combining the exact
inhomogeneous solutions found by Lim~\cite{lim08} into \emph{concatenation
blocks} that connect families of timelines affected by spike formation at
common initial and final Kasner states. It was also shown that the
concatenation blocks come in two versions, so-called \emph{high velocity
transitions}, $\Thi$, and \emph{joint low/high velocity spike transitions},
$\Tco$, where the word transition, as in the BKL case, signifies a change
from one Kasner state to another. Hence both BKL dynamics and spike dynamics
can be described in terms of temporally oscillating Kasner states, and these
oscillations can be described by discrete maps.

The outline of the paper is as follows. In the next section we describe the
BKL and spike maps, and the maps they in turn give rise to, and some of these
maps properties; due to the central feature that all these maps describe
changes in Kasner states, we refer to them collectively as \emph{Kasner
maps}. Furthermore, we derive a distribution function that allows one to
statistically analyze properties associated with spike induced Kasner
sequences, which complements previous corresponding BKL results. In
Section~\ref{sec:statistics} we explore some statistical features the maps
give rise to. In particular we prove that statistically so-called small
curvature phases dominate over large curvature phases, and we also give
stochastical relations for the `Hubble-normalized' Weyl scalar, for generic
BKL as well as spike induced Kasner sequences. We conclude the paper with a
few remarks about future possible research projects concerning explorations
of chaotic features of recurring spike formation, which we leave for the
interested reader.

\section{Kasner maps}\label{sec:maps}

A Kasner state can be described in terms of the gauge invariant \emph{Kasner
parameter} $u$, see e.g.~\cite{bkl70,heietal09,ugg13,khaetal85}.\footnote{The
1-parameter family of Kasner solutions is often given in terms of the line
element $ds^2 = -dt^2 + t^{2p_1}dx^2 + t^{2p_2}dy^2 + t^{2p_3}dz^2$, where
$p_1+p_2+p_3=p_1^2+p_2^2+p_3^2 = 1$. The Kasner parameter $u$ can be defined
via $p_1p_2p_3 = -u^2(1+u)^2/(1+u+u^2)^3$.} Next we describe the two types of
discrete maps which determine the change in $u$ due to the transitions,
beginning with the map that arises from the Bianchi type II solutions (i.e.,
single curvature transitions) in the BKL scenario. Throughout, the time
direction is chosen to be toward the initial singularity.

\subsection{The BKL map}

The single curvature transitions induce the following \emph{BKL
map}~\cite{bkl70,khaetal85}:
\begin{equation}\label{BKLMap}
u^{\mathrm{f}}  \:= \: \left\{\begin{array}{ll}
u^{\mathrm{i}} - 1 & \qquad \text{if}\quad u^{\mathrm{i}} \in [2, \infty)\:, \\[1ex]
(u^{\mathrm{i}} - 1)^{-1} &\qquad  \text{if} \quad u^{\mathrm{i}} \in [1,2]\:.
\end{array}\right.
\end{equation}
This map is often referred to as the Kasner map, but since we here will use
the term Kasner map as an umbrella term for maps that induce changes in a Kasner
state, we instead refer to this particular Kasner map as the BKL map.

\subsection{The BKL era map}

Iterations of the BKL map generate, from every initial value $u_0 \in
[1,\infty)$, a (finite or infinite) BKL \textit{sequence of Kasner epochs}
$(u_l)_{l=0,1,2,\ldots}$. We say that the Kasner epochs $u_l$ and $u_{l+1}$
belong to the same (BKL) \textit{era} if $u_{l+1} = u_l - 1$, thereby leading
to a partition of a BKL sequence of Kasner epochs into
eras~\cite{bkl70,khaetal85}, e.g.,
\begin{small}
\begin{equation*}
\ldots \rightarrow 1.14 \rightarrow \underbrace{7.29 \rightarrow 6.29 \rightarrow 5.29 \rightarrow 4.29 \rightarrow 3.29
\rightarrow 2.29 \rightarrow 1.29}_{\text{\scriptsize era}} \rightarrow
\underbrace{3.45 \rightarrow 2.45 \rightarrow 1.45}_{\text{\scriptsize  era}}  \rightarrow
\underbrace{2.24 \rightarrow 1.24}_{\text{\scriptsize era}}  \rightarrow \ldots
\end{equation*}
\end{small}

Let us denote the initial ($=$ maximal) value of the Kasner parameter $u$ in
era number $s$ by $\u_s$, where $s = 0,1,2,\ldots\:\:\:$.
Following~\cite{bkl70,khaetal85}, we decompose $\u_s$ into its integer part
$k_s = [\u_s]$ and its fractional part $\{\u_s\}$, i.e.,
\begin{equation}\label{usdecomp}
\u_s = k_s + \{\u_s\} \:.
\end{equation}
The number $k_s$ represents the (discrete) length of era $s$, which is simply
the number of Kasner epochs it contains. The final ($=$ minimal) value of the
Kasner parameter in era $s$ is given by $1 + \{\u_s\}$, which implies that era
number $(s+1)$ begins with
\begin{equation}
\u_{s+1} = \frac{1}{\{\u_s\}}\:.
\end{equation}
The map $\u_s \mapsto \u_{s+1}$ is (a variant of) the \textit{BKL era map},
which we usually, for brevity, refer to as just the era map; starting from
$\u_0 = u_0$ it recursively determines $\u_s$, $s=0,1,2,\ldots$, and thereby
the complete BKL Kasner sequence $(u_l)_{l=0,1,\ldots}$ of epochs.

The era map admits a straightforward interpretation in terms of continued
fractions. Consider the continued fraction representation of the initial
value, i.e.,
\begin{equation}
\u_0 =  k_0 + \cfrac{1}{k_1 + \cfrac{1}{k_2 + \dotsb}} = [k_0; k_1,k_2,k_3,\dotsc]\:.
\end{equation}
The fractional part of $\u_0$ is $[0;k_1,k_2,k_3,\dotsc]$; since $\u_1$ is
the reciprocal of $\{\u_0\}$, it follows that
\begin{equation}
\u_1 = [k_1; k_2,k_3,k_4,\dotsc]\:.
\end{equation}
Therefore, the era map is simply a shift to the left in the continued
fraction expansion,
\begin{equation}
\u_s = [k_s; k_{s+1}, k_{s+2}, \dotsc] \:\mapsto\:
\u_{s+1} = [k_{s+1}; k_{s+2}, k_{s+3},\dotsc]\:.
\end{equation}

In terms of continued fractions, the Kasner sequence $(u_l)_{l\in\mathbb{R}}$
generated by $u_0 =  \big[ k_0; k_1, k_2 , \dotsc \big]$ is given by
\begin{align*}
u_0 = \mathsf{u}_0 & =
\big[ k_0; k_1, k_2,  \dotsc \big] \rightarrow  \big[ k_0 -1 ; k_1, k_2 , \dotsc \big]
\rightarrow \big[ k_0 -2 ; k_1, k_2 ,  \dotsc \big]
\rightarrow \ldots \rightarrow \big[1 ; k_1, k_2 , \dotsc \big] \\
 \rightarrow \mathsf{u}_1 & = \big[ k_1; k_2 , k_3, \dotsc \big]
 \rightarrow  \big[  k_1-1; k_2 , k_3, \dotsc \big]
\rightarrow \big[ k_1 -2 ; k_2 , k_3, \dotsc \big]
\rightarrow \ldots \rightarrow \big[1 ; k_2, k_3 , \dotsc \big] \\
 \rightarrow \mathsf{u}_2 & = \big[ k_2; k_3 , k_4, \dotsc \big]
 \rightarrow  \big[  k_2-1; k_3 , k_4, \dotsc \big]
\rightarrow \big[ k_2 -2 ; k_3 , k_4, \dotsc \big] \rightarrow \ldots\:\,.
\end{align*}

Since $\u_s \in [k_s, k_s+1)$, the number $k_s$ represents the (discrete)
length of era $s$, which is simply the number of Kasner epochs it contains.
Therefore, passing on to the stochastical interpretation of (generic) Kasner
sequences of epochs, we find that the probability that a randomly chosen era
$s$ of a Kasner sequence $(u_l)_{l=0,1,2,\ldots}$ of epochs has length $m
\in\mathbb{N}$ corresponds to the probability that $k_s = m$, or,
equivalently, to the probability that $\u_s \in [m, m+1)$. Since the sequence
$(k_0, k_1, k_2, \dotsc)$ arises as the continued fraction expansion of $u_0$
this probability in turn corresponds to the probability that a randomly
chosen partial quotient in the
continued fraction expansion 
is equal to $m$. Hence we resort to Khinchin's law~\cite{kin64}, which states
that the partial quotients of the continued fraction representation of a
generic real number are distributed like a random variable whose probability
distribution is given by
\begin{equation}\label{Khinchin}
\kh(m)   \,= \,  \log_2 \left( \frac{m+1}{m+2}\right) -\log_2 \left( \frac{m}{m+1} \right)\,,
\end{equation}
which leads to
\begin{equation}\label{lengthBKL}
\mathrm{Probability}\big(\text{length of era} = n\big) = \khl(n) = \kh(n)\,.
\end{equation}
It follows that $42\%$ of the BKL eras of a generic BKL Kasner sequence of
epochs contain merely one Kasner epoch; $17\%$ of the eras contain $2$ Kasner
epochs, and, e.g., $1.4*10^{-2}\%$ of the eras consist of $100$ Kasner
epochs.

\subsection{The spike map}

Let us now consider the Kasner map that is induced by the spike transitions
$\Thi$ and $\Tco$. It is shown in~\cite{heietal12} that, remarkably, the two
types of transitions give rise to the same Kasner map which we refer to as
the \emph{spike map}. This map is obtained by applying the BKL map twice,
see~\cite{heietal12}, which results in the following relations:
\begin{equation}\label{spikemap}
u_+ =
\begin{cases}
u_- - 2 & u_- \in [3 ,\infty)\:, \\
(u_- -2)^{-1} & u_- \in [2,3] \:,\\
\big((u_- -1)^{-1} - 1 \big)^{-1} & u_- \in [ 3/2 ,2] \:,\\
(u_- -1)^{-1} - 1 & u_- \in [1, 3/2]\:.
\end{cases}
\end{equation}
%

\subsection{The spike era map}

Asymptotically we expect that spike dynamics is reduced to dynamics along
timelines associated with a spatial `spike' surface. Furthermore, this gives
rise to recurring spikes described by sequences of $\Thi$ and $\Tco$, which
yield iterations of the spike map~\cite{heietal12}. Iterations of the spike
map~\eqref{spikemap} generate, from every initial value $u_0 \in [1,\infty)$,
a (finite or infinite) recurring spike-generated sequence of Kasner epochs
$(u_l)_{l=0,1,2,\ldots}$. We make a partition of $(u_l)_{l=0,1,2,\ldots}$
into recurring spike-induced eras, which we denote as spike eras, or, for
brevity, $\era$s: $u_l$ and $u_{l+1}$ belong to the same $\era$ if $u_{l+1} =
u_l - 2$. If $u_{l+1}$ arises from $u_l$ by one of the other three laws
of~\eqref{spikemap}, we speak of a change of $\era$.
\begin{equation*}
\ldots \rightarrow 1.14 \rightarrow \underbrace{7.29 \rightarrow 5.29
\rightarrow 3.29 \rightarrow 1.29}_{\text{\scriptsize $\era$}}
\rightarrow \underbrace{2.45}_{\text{\scriptsize $\era$}} \rightarrow
\underbrace{2.24}_{\text{\scriptsize $\era$}} \rightarrow
\underbrace{4.16 \rightarrow 2.16}_{\text{\scriptsize $\era$}}  \rightarrow
\underbrace{6.14  \rightarrow 4.14 \rightarrow 2.14}_{\text{\scriptsize $\era$}}  \rightarrow \ldots
\end{equation*}
We denote the initial ($=$ maximal) value of the Kasner parameter $u$ in
$\era$ number $s$ (where $s = 0,1,2,\ldots$) by $\u_s$. The spike map induces
an $\era$ map $\u_s \mapsto \u_{s+1}$, which recursively determines
$(\u_s)_{s\in\mathbb{N}}$ from $\u_0 = u_0$, and thereby the complete spike
induced sequence $(u_l)_{l=0,1,2,\ldots}$ of Kasner epochs.

The length of an $\era$ $s$ is determined by the value of $\u_s$: If $\u_s
\in [m,m+1)$ for some $m\in\mathbb{N}$, then the length of the $\era$ is
$m/2$, if $m$ is even, and $(m+1)/2$, if $m$ is odd. In the stochastical
framework, in analogy with~\eqref{Khinchin}, let $\khp(m)$ denote the
probability that a randomly chosen element 
of an $\era$ sequence $(\u_s)_{s=0,1,2,\ldots}$ lies in the interval
$[m,m+1)$.
\begin{theorem}\label{nonKhinchinlemma}
Let $(\u_s)_{s\in\mathbb{N}}$ be a generic spike-induced sequence of $\era$s;
then the probability that a randomly chosen element of
$(\u_s)_{s\in\mathbb{N}}$ lies in the interval $[m,m+1)$ is
\begin{equation}\label{Khinchinp}
\khp(m)   \,= \,  \log_3 \left( \frac{m+2}{m+3}\right) -\log_3 \left( \frac{m}{m+1} \right) \,.
\end{equation}
\end{theorem}

\begin{proof}
We perform the proof in three steps.

\underline{Step 1}: Representation of the $\era$ map $\u_s \mapsto \u_{s+1}$
by means of continued fractions. The spike map~\eqref{spikemap} yields the
following respresentations in terms of continued fraction expansions of $u_\pm$
for the four different cases of~\eqref{spikemap}:
\begin{subequations}\label{spikemapcontfrac}
\begin{alignat}{2}
\label{spikemapcontfrac1}
u_- & = [2 + k; k_1, k_2, k_3, \ldots ] & \:& \mapsto\: u_+ = [k; k_1, k_2, k_3, \ldots ]\,, \\
\label{spikemapcontfrac2}
u_- & = [2; k_1, k_2, k_3, k_4, \ldots ] & & \mapsto \:u_+ = [k_1; k_2, k_3, \ldots ] \,,\\
\label{spikemapcontfrac3}
u_- & = [1; 1, k_2, k_3, k_4, \ldots\, ] & & \mapsto\: u_+ = [k_2; k_3, k_4, \ldots ] \,,\\
\label{spikemapcontfrac4} u_- & = [1; 1 + k, k_2, k_3, \ldots ] & & \mapsto
\:u_+ = [k; k_2, k_3, k_4, \ldots ]\:,
\end{alignat}
\end{subequations}
where $k$ and $k_i$, $i=1,2,\ldots$, are positive natural numbers. Recall
that~\eqref{spikemapcontfrac2}--\eqref{spikemapcontfrac4} describe a change
of $\era$. The spike map~\eqref{spikemapcontfrac} leads to a representation
of the $\era$ map $\u_s \mapsto \u_{s +1}$: If
\begin{align}
\nonumber
\u_s & = [\varkappa_s; \varkappa_{s+1}, \varkappa_{s+2}, \ldots] \,,
\intertext{then}
\label{eramap}
\u_{s+1} & =
\begin{cases}
\big[\varkappa_{s+1} - (\varkappa_s \,\mZ); \varkappa_{s+2}, \varkappa_{s+3}, \ldots\big]
\quad & \text{if}\quad\varkappa_{s+1} - (\varkappa_s \,\mZ) \neq 0\,,\\[0.5ex]
[\varkappa_{s+2}; \varkappa_{s+3}, \varkappa_{s+4}, \ldots]  & \text{if}\quad\varkappa_{s+1} - (\varkappa_s \,\mZ) = 0\,,
\end{cases}
\end{align}
when we use the standard notation $(\varkappa\, \mZ)$ for natural numbers
(where we recall that $(\varkappa\, \mZ) = 0$ if $\varkappa$ is even;
$(\varkappa\, \mZ) = 1$ if $\varkappa$ is odd).

\underline{Step 2}: Iterate~\eqref{eramap} and represent the elements $\u_s$,
$s=0,1,2,\ldots$, of the $\era$ sequence in terms of the partial quotients of
the continued fraction representation of $\u_0$. Let $\u_0$ be given by $\u_0
= [k_0; k_1, k_2,\ldots]$. We begin by constructing from $(k_0, k_1,
k_2,\ldots)$ an auxiliary sequence $(\kappa_0',\kappa_1',\kappa_2',\ldots)$
according to $\kappa_0' = k_0$ and
\begin{equation}\label{kkappa}
\kappa_{s+1}^\prime = k_{s+1} - (\kappa_{s}^\prime\:\mZ) =
k_{s+1} - \Big(\big[\sum\nolimits_{i= 0}^{s} k_i\big] \:\mZ\Big)\,.
\end{equation}
In the generic case, the number $0$ will appear (infinitely often) in
$(\kappa_0',\kappa_1',\kappa_2',\ldots)$. We remove these entries and denote
the arising (sub)sequence by $(k_0',k_1',k_2',\ldots)$. A typical example is
\begin{alignat*}{3}
(k_0, k_1, k_2,\ldots) & = (4,3,6,5,1,2,1,\, &&7,2,7,1,1,\,&& 3,\ldots)\,,   \\
(\kappa_0',\kappa_1',\kappa_2',\ldots) & = (4,3,5,4,1,1,0, \,&& 7,1,6,1,0, \,&& 3,\ldots) \,,\\
(k_0',k_1',k_2',\ldots)  & = (4,3,5,4,1,1, & & 7,1,6,1,& & 3,\ldots) \,.
\end{alignat*}
Accordingly, we have
\begin{equation}\label{kappaandkprime}
k_{s}' =\kappa_{s+\z_s}^\prime = k_{s + \z_s} - \Big(\big[\sum\nolimits_{i= 0}^{s + \z_s -1} k_i\big] \:\mZ \Big) \,,
\end{equation}
where $\z_s$ is the number of slots (i.e., zeros) that have been removed up
to index number $s$ of the sequence $(k_0',k_1',k_2',\ldots)$.

Comparing~\eqref{eramap} and~\eqref{kkappa} we find that the $\era$ sequence
reads
\[
\u_0 = [\kappa_0'; k_1, k_2, \ldots], \quad \u_1 = [\kappa_1'; k_2, k_3, \ldots],\quad
\u_2 = [\kappa_2'; k_3, k_4, \ldots], \ldots
\]
up the smallest index $s$ for which $k_{s+1} - (\kappa_{s}^\prime\:\mZ) = 0$;
at that point, the second case of~\eqref{eramap} becomes relevant, which
leads to an omission of $\kappa_{s+1}^\prime$ ($ = 0$). The appropriate
sequence that enters our representation of $(\u_s)_{s\in\mathbb{N}}$ is thus
the constructed (sub)sequence $(k_0',k_1',k_2',\ldots)$: Making use of
$(k_0',k_1',k_2',\ldots)$ we are able to give the element $\u_s$ of the
$\era$ sequence as
%
\begin{equation}\label{uspike}
\u_s  = [\kappa_{s+\z_s}^\prime; k_{s+\z_s+ 1}, k_{s+\z_s+ 2},\ldots] = [k_s^\prime; k_{s+\z_s+ 1}, k_{s+\z_s+ 2},\ldots] \,,
\end{equation}
where $\z_s$ is the number of omissions up to index number $s$. It follows
that the probability $\khp(m)$ that a randomly chosen element of
$(\u_s)_{s\in\mathbb{N}}$ lies in the interval $[m,m+1)$ is identical to the
probability that a randomly chosen element of $(k_0',k_1',k_2',\ldots)$ is
equal to $m$.

\underline{Step 3}: Here we use the stochastical properties of sequences
$(k_0',k_1',k_2',\ldots)$ that arise from (generic) sequences
$(k_0,k_1,k_2,\ldots)$ to derive the desired probability distribution. As a
preliminary, we consider sums of the type $\sum_{i=0}^j k_i$, where $(k_0,
k_1, k_2, \ldots)$ are randomly distributed natural numbers, e.g., according
to Khinchin's law~\eqref{Khinchin}. Let $0 < p < 1$ be the probability for an
element of $(k_0, k_1, k_2, \ldots)$ to be odd (which is $\sum_{m=1,
m\:\text{odd}}^\infty \kh(m)$ in the case of Khinchin's law). Note that $p  =
1/2$ is not needed; in fact $p \approx 0.6515$ in the case of of Khinchin's
law. The sum $\sum_{i=0}^j k_i$ is even or odd, if the number of odd summands
is even or odd, respectively. Hence, making use of the binomial distribution,
we find that the probability for $\sum_{i=0}^j k_i$ to be even or odd is
\[
\sum_{i=0, i \;\text{even}}^{j} \binom{j}{i}\, p^i (1-p)^{j-i}\,,\qquad
\sum_{i=0, i \;\text{odd}}^{j} \binom{j}{i}\, p^i (1-p)^{j-i}\,,
\]
respectively. The two probabilities are not exactly $1/2$, but they converge
to $1/2$ rapidly as $j$ grows, and hence the two probabilities are equal to
$1/2$ in the asymptotic limit of very long sequences.\footnote{For $j=10$,
the error is already of the order of $10^{-6}$ in the case of Khinchin's law;
for a mere $100$ elements, the error is already less than $10^{-50}$.}

Finally, let $\u_0 = [k_0; k_1, k_2,\ldots]$ be a generic real number.
Consider the truncated sequence $(k_0, k_1,\dotsc,k_n)$, where $n \gg 1$, and
denote the probability distribution of its elements by $\kh_n$. In the limit
$n\rightarrow \infty$, $\kh_n$ converges to Khinchin's distribution $\kh$,
cf.~\eqref{Khinchin}. In the auxiliary sequence $(\kappa_0',
\kappa_1',\dotsc, \kappa_n')$ constructed from $(k_0, k_1,\dotsc,k_n)$, the
number zero appears in approximately $21\%$ of the slots. This is because
of~\eqref{kkappa}: $\kappa'_s$ is $0$ if and only if $k_s$ equals $1$ and
$\sum_{i=1}^{s-1} k_i$ is odd. The probability of the former event is
$\kh_n(1) \approx \kh(1)$, cf.~\eqref{Khinchin}, which is approximately
$42\%$ (because $n\gg 1$); the probability of the latter event is $1/2$ in
the limit $n\rightarrow \infty$. Since the two events are independent, we
obtain a probability of $\kh_n(1)/2 \approx \kh(1)/2 \approx 21\%$ that the
number $0$ appears in a randomly chosen slot of $(\kappa_0',\kappa_1',\dotsc,
\kappa_n')$. Since the (sub)sequence $(k_0',k_1',\dotsc, k_{\bar{n}}')$ is
constructed from $(\kappa_0',\kappa_1',\dotsc, \kappa_n')$ by removing each
appearance of the number~$0$, $(k_0',k_1',\dotsc,k_{\bar{n}}')$ is shorter
than $(\kappa_0',\kappa_1',\dotsc, \kappa_n')$ by $21\%$, i.e., $\bar{n}
\approx 0.79\, n$. We are thus able to compute the probability
$\khp_{\bar{n}}(m)$ that a randomly chosen element $k_s'$ of the sequence
$(k'_0, k'_1,\dotsc,k_{\bar{n}}')$ is equal to $m$. There are two ways that
the event $k_s' = m$ can occur. Either the corresponding element
$k_{s+\z_s}$, from which $k_s'$ is generated, see~\eqref{kappaandkprime}, is
equal to $m$ and $\sum_{i=0}^{s+ \z_s-1} k_i$ is even, or it is equal to
$m+1$ and $\sum_{i = 0}^{s+ \z_s-1} k_i$ is odd. Accordingly,
\begin{equation}
\khp_{\bar{n}}(m) = \frac{ \frac{1}{2} \,\kh_n(m) + \frac{1}{2}\,\kh_n(m+1)}{1 - \frac{1}{2}\,\kh_n(1)} \,.
\end{equation}
The limit $\bar{n} \rightarrow \infty$ corresponds to $n\rightarrow \infty$;
since $\kh_n(m) \rightarrow \kh(m)$ in this limit,~\eqref{Khinchin} implies
\begin{equation}
\khp(m) = \frac{ \frac{1}{2} \,\kh(m) + \frac{1}{2}\,\kh(m+1)}{1 - \frac{1}{2}\,\kh(1)}=
\log_3 \left( \frac{m+2}{m+3}\right) -\log_3 \left( \frac{m}{m+1} \right) \:,
\end{equation}
which concludes the proof.
\end{proof}


\begin{corollary}
Let $(\u_s)_{s\in\mathbb{N}}$ be a generic spike induced sequence of $\era$s;
then the probability that a randomly chosen $\era$ in this sequence possesses
length $n$ is given by
\begin{equation}\label{lengthnonBKL}
\mathrm{Probability}\left(\text{length of $\era$} = n\right) =: \khpl(n) =
\log_3 \left( \frac{2 n+ 1}{2 n+3}\right) -\log_3 \left( \frac{2 n-1}{2 n+1} \right)\,.
\end{equation}
\end{corollary}
\begin{proof}
This is a direct consequence of Theorem~\ref{nonKhinchinlemma}: The length of
$\era$ $s$, where $\u_s  = [\varkappa_s; \varkappa_{s+1},\varkappa_{s+2},
\dotsc]$, is $\varkappa_s/2$ if $\varkappa_s$ is even and $(\varkappa_s+1)/2$
if $\varkappa_s$ is odd. Therefore, for $\era$ $s$ to be of length $n$,
$\varkappa_s$ must be either $2 n$ or $2 n -1$. From~\eqref{uspike}
and~\eqref{Khinchinp} we see that the probability that one of these events
occur is $\khp(2n)+ \khp(2n-1)$, which yields~\eqref{lengthnonBKL}.
\end{proof}
%

%
%
%
%

\section{Statistics}\label{sec:statistics}

\subsection{Comparison between BKL eras and spike eras}

Let us compare some consequences of the probability
distribution~\eqref{Khinchin}, which determines the probabilities for
prescribed lengths of BKL eras in BKL sequences of Kasner epochs according to
eq.~\eqref{lengthBKL}, and the probability distribution~\eqref{Khinchinp},
which determines the probabilities for prescribed lengths of $\era$s in
spike-induced sequences of Kasner epochs according to
eq.~\eqref{lengthnonBKL}, see Tables~\ref{tab:probcomp1}
and~\ref{tab:probcomp2}. As seen in Table~\ref{tab:probcomp2}, $\era$s have
the tendency of being shorter than BKL eras. The probability that an $\era$
contains merely one epoch is larger than $50\%$; the probability that an
$\era$ consists of $n > 1$ epochs is smaller than that of a BKL era.
Asymptotically, for $n \gg 1$, we have
\begin{alignat*}{2}
&\mathrm{Probability}\left(\text{length of era} = n\right) \,& & =
\, (\log 2)^{-1} \, n^{-2}\:\left( 1  - 2 n^{-1} + O(n^{-2}) \right) \,,\\
&\mathrm{Probability}\left(\text{length of $\era$} = n\right) \,& &= \, (\log
3)^{-1} \,n^{-2} \:\left(1 -  n^{-1} + O(n^{-2}) \right) \,,
\end{alignat*}
and hence the two probabilities are asymptotically proportional with a
proportionality factor $\log 2/\log 3$.

\begin{table}[Htp]
\begin{center}
\begin{tabular}{|cccccccccc|}
\hline Sequence   & 1 & 2 & 3 & 4 & 5 & 10 & 50 & 100 & 500 \\ \hline 
era &  $41.50$ & $16.99$ & $9.31$ & $5.89$ & $4.06$ &  $1.20$ & $5.5\times10^{-2}$
& $1.4\times10^{-2}$ & $5.7\times10^{-4}$  \\ 
$\era$ & $36.91$ & $16.60$ & $9.60$ & $6.28$ & $4.44$
& $1.39$ & $6.9\times10^{-2}$ & $1.8 \times10^{-2}$ & $7.2\times10^{-4}$
\\ \hline 
\end{tabular}
\caption{Probabilities (in $\%$) that a randomly chosen element of a
BKL/spike-generated Kasner sequence of Kasner epochs
$(\u_s)_{s\in\mathbb{N}}$, is in the interval $[m,m+1)$, $m = 1,2,3,\ldots$.
These are the probability distributions $\kh(m)$ and $\khp(m)$,
see~\eqref{Khinchin} and~\eqref{Khinchinp}.} \label{tab:probcomp1}
\end{center}
\end{table}

\begin{table}[Htp]
\begin{center}
\begin{tabular}{|cccccccccc|}
\hline Length  & 1 & 2 & 3 & 4 & 5 & 10 & 50 & 100 & 500 \\ \hline 
era &  $41.50$ & $16.99$ & $9.31$ & $5.89$ & $4.06$ &  $1.20$ & $5.5\times10^{-2}$
& $1.4\times10^{-2}$ & $5.7\times10^{-4}$  \\ 
$\era$  & $53.50$ & $15.87$ & $7.75$ & $4.61$ & $3.06$ & $0.83$ & $3.5\times10^{-2}$
& $0.9\times10^{-2}$ & $3.6\times10^{-4}$ \\ \hline 
\end{tabular}
\caption{Probabilities (in $\%$) that a randomly chosen era/$\era$ of a
BKL/spike generated sequence of Kasner epochs, is of a prescribed length.
These are the probability distributions $\khl(m) =\kh(m)$ and $\khpl(m)$,
see~\eqref{lengthBKL} and~\eqref{lengthnonBKL}.} \label{tab:probcomp2}
\end{center}
\end{table}
%

\subsection{Dominance of small curvature phases}

Following~\cite{heietal09} we first introduce $\Upsilon > 3$, where we are
foremost interested in the case $\Upsilon\gg 1$. Then we define a
\textit{small curvature phase} of a BKL or spike induced sequence of Kasner
epochs $(u_l)_{l\in\mathbb{N}}$ as a connected and inextendible piece
$\mathcal{L} \subset \mathbb{N}$ such that $u_l > \Upsilon$ $\forall l\in
\mathcal{L}$. During a small curvature phase $u_l$ is thus monotonically
decreasing from a maximal value by the BKL map to a minimal value in the
interval $(\Upsilon, \Upsilon+1]$, while the spike map yields a monotonic
decrease from a maximum value to a minimal value in the interval $(\Upsilon,
\Upsilon+2]$. The complement of the concept of a small curvature phase is a
\textit{large curvature phase}, which is defined as an (inextendible) piece
of the sequence of Kasner epochs such that $u_l \leq \Upsilon$ for all $l$.

During small/large curvature phases, the `Hubble-normalized' curvature is
comparatively small/large, see eq.~\eqref{weylu} below. While a small
curvature phase can be viewed as an era/$\era$ that is terminated prematurely
at $\Upsilon$, a large curvature phase typically consists of many
eras/$\era$s; clearly, small and large curvature phases occur alternately. In
the following BKL example, where the choice $\Upsilon = 4$ has been made, the
large curvature phase contains two and a half eras.
\begin{small}
\begin{equation*}
\ldots \rightarrow 1.14 \rightarrow \overunderbraces{& \br{1}{\text{\scriptsize small curvature phase}}
& & \br{5}{\text{\scriptsize large curvature phase}}}%
{& 7.29 \rightarrow 6.29 \rightarrow 5.29 \rightarrow 4.29 & \rightarrow & 3.29
\rightarrow 2.29 \rightarrow 1.29 & \rightarrow & 3.45 \rightarrow
2.45 \rightarrow 1.45 & \rightarrow & 2.24
\rightarrow 1.24 & \rightarrow  & \ldots}%
{ &\br{3}{\text{\scriptsize era}} & & \br{1}{\text{\scriptsize era}}
& & \br{1}{\text{\scriptsize era}} & }
\end{equation*}
\end{small}

Combining the probabilistic viewpoint with the concept of small/large
curvature phases lead to a fundamental result in the description of the BKL
and spike induced Kasner sequences. For generic Kasner sequences
$(u_l)_{l\in\mathbb{N}}$, \textit{small curvature phases dominate over large
curvature phases} in the following sense:
\begin{theorem}\label{domlemma}
Let $(u_l)_{l\in\mathbb{N}}$ be a generic BKL or spike-induced Kasner
sequence and let $\Upsilon$ be arbitrarily large. Then for a randomly chosen
epoch $u$ the probability for the event $u > \Upsilon$ is one and the
probability for the event $u \leq \Upsilon$ is zero.
\end{theorem}
%
%
\begin{proof}
Let us first consider the BKL case and let us for simplicity choose $\Upsilon
\in \mathbb{N}$. We let $u_0 = \big[ k_0; k_1, k_2 , \dotsc \big]$ be a well
approximable irrational number and consider the BKL Kasner sequence of epochs
$(u_l)_{l\in\mathbb{N}}$ and the associated era sequence
$(\u_s)_{s\in\mathbb{N}}$ (where we recall
$\u_s = \big[ k_s; k_{s+1},\dotsc \big]$). Consider the truncated sequence of
eras $(\u_s)_{s\leq n}$ and denote the associated probability distribution by
$\kh_n$ (which converges to Khinchin's law $\kh$ in the limit $n\rightarrow
\infty$). Let $u$ be a randomly chosen element from the associated (finite)
sequence of epochs $(u_l)_{l\leq k_0 + \cdots + k_n}$. We denote by
$\mathcal{E}_m$ the collection of epoch of eras of length $m$ and by
$\tilde{P}_n(u{\in}\mathcal{E}_m)$ the probability that the epoch $u$ belongs
to an era of length $m$; it is given by
\begin{equation}
\tilde{P}_n(u{\in}\mathcal{E}_m) = \Big( \sum_{m'} m' \,\kh_n(m') \Big)^{-1} \: m\, \kh_n(m) \:.
\end{equation}
The probability for $u \leq \Upsilon$, which we denote by
$P_n(u{\leq}\Upsilon)$, is thus obtained from
\begin{align*}
P_n\big( u{\leq}\Upsilon\big) & =
\sum_{m=1}^{\Upsilon-1} \tilde{P}_n(u{\in}\mathcal{E}_m) +
\sum_{m=\Upsilon}^\infty  \tilde{P}_n(u{\in}\mathcal{E}_m) \frac{\Upsilon-1}{m} \\
& = \Big( \sum_{m'} m'\, \kh_n(m') \Big)^{-1}
\left[ \sum_{m=1}^{\Upsilon-1}  m\, \kh_n(m) +
(\Upsilon-1) \sum_{m=\Upsilon}^\infty   \kh_n(m) \right]\:.
\end{align*}
For a generic number $u_0$, in the asymptotic limit $n\rightarrow \infty$,
the probability $\kh_n(m)$ converges to $\kh(m)$ given by Khinchin's
law~\eqref{Khinchin}. Since
\begin{equation}\label{Pseri}
\kh(m) =  
\frac{1}{\log 2} \:\left[ m^{-2} -2 m^{-3} + O(m^{-4}) \right]
\qquad (m\rightarrow\infty)\:,
\end{equation}
the sum $\sum_{m=\Upsilon}^\infty \kh_n(m)$ approaches the convergent series
$\sum_{m=\Upsilon}^\infty \kh(m)$ in the limit $n\rightarrow \infty$;
however, the sum $\sum_{m} m \kh_n(m)$ diverges as $n\rightarrow \infty$, so
that
\begin{equation}\label{Piszero}
P\big( u{\leq}\Upsilon\big) = \lim_{n\rightarrow\infty} P_n\big( u{\leq}\Upsilon\big) = 0\:,
\end{equation}
irrespective of the value of $\Upsilon$. We conclude that the probability
that an epoch $u$, which is randomly chosen from the Kasner sequence
$(u_l)_{l\in \mathbb{N}}$, belongs to a large curvature phase ($u\leq
\Upsilon$) is zero; with probability $1$ a randomly chosen element belongs to
a small curvature phase ($u> \Upsilon$).

It remains to consider the recurring spike case, where we choose an
arbitrarily large $\Upsilon$ such that $\Upsilon = 2\upsilon +1$ with
$\upsilon \in \mathbb{N}$ for simplicity. Let $\khpl_n(m)$ denote the
probability for an $\era$ of the truncated sequence to be of length $m$;
$\khpl_n(m)$ converges to $\khpl(m)$ as $n\rightarrow \infty$,
cf.~\eqref{lengthnonBKL}. Then
\begin{equation}
\tilde{P}_n(u{\in}\mathcal{E}_m) = \Big( \sum_{m'} m' \,\khpl_n(m') \Big)^{-1} \: m\, \khpl_n(m) \:,
\end{equation}
and the probability for $u \leq \Upsilon$ reads
\begin{equation*}
P_n\big( u{\leq}\Upsilon\big)  =
\sum_{m=1}^{\upsilon} \tilde{P}_n(u{\in}\mathcal{E}_m) +
\sum_{m=\upsilon+1}^\infty  \tilde{P}_n(u{\in}\mathcal{E}_m) \frac{\upsilon}{m} =
\frac{\sum_{m=1}^{\upsilon}  m\, \khpl_n(m) +
\upsilon \sum_{m=\upsilon+1}^\infty   \khpl_n(m)}{\sum_{m'} m'\, \khpl_n(m')}\:.
\end{equation*}
In the case of a generic sequence, $\khpl_n(m)$ converges to $\khpl(m)$ as
$n\rightarrow \infty$, where
\begin{equation}\label{Pseri2}
\khpl(m) =  
\frac{1}{\log 3} \:\left[ m^{-2} - m^{-3} + O(m^{-4}) \right]
\qquad (m\rightarrow\infty)\:.
\end{equation}
Hence the sum $\sum_{m=\upsilon+1}^\infty \khpl_n(m)$ converges as
$n\rightarrow \infty$. However, the sum $\sum_{m} m \khpl_n(m)$ diverges as
$n\rightarrow \infty$ and $P\big( u{\leq}\Upsilon\big) =
\lim_{n\rightarrow\infty} P_n\big( u{\leq}\Upsilon\big) = 0$.
\end{proof}

The underlying reason for the dominance of small curvature phases is hence
the failure of the probability
distributions~\eqref{Khinchin},~\eqref{Khinchinp}
(and~\eqref{lengthBKL},~\eqref{lengthnonBKL}), to generate finite expectation
values, since
\begin{equation}\label{sumdiv}
\sum\limits_{m=1}^\infty m \,\kh(m) = \sum\limits_{m=1}^\infty m \,\khl(m) = \infty \:,\qquad
\sum\limits_{m=1}^\infty m \,\khp(m) = \infty =\sum\limits_{m=1}^\infty m \,\khpl(m) \:,
\end{equation}
which is due to the infinite tail of the distributions. Accordingly, the
average length of an era/$\era$ is ill-defined.

\subsection{Weyl stochastics}\label{Weylstoch}

The magnetic part of the Weyl tensor for any Kasner state is zero, and hence
the Hubble-normalized Weyl scalar is given by $\mathcal{W}_{\mathrm{K}}^2 =
\sfrac12 H^{-4} E_{\alpha\beta}E^{\alpha\beta}$, where $E_{\alpha\beta}$ is
the electric Weyl tensor and $H$ is the Hubble variable, which is related to
the expansion $\theta$ of the timelike reference congruence according to
$H=\theta/3$, see~\cite{heietal09,ugg13}. Expressed in terms of the Kasner
parameter $u$, $\mathcal{W}_{\mathrm{K}}^2$ is given by
\begin{equation}\label{weylu}
\mathcal{W}_{\mathrm{K}}^2 = -81p_1p_2p_3 = \frac{81 u^2 (1+u)^2}{(1+u+u^2)^3}\:.
\end{equation}
The behavior of $\mathcal{W}_{\mathrm{K}}^2$ over generic sequences
$(u_l)_{l\in\mathbb{N}}$ of Kasner epochs is statistically described by the
following Theorem:
\begin{theorem}\label{theoremW}
Let $\u_1 = u_1 = [k_1;k_2,k_3, \ldots]$ be a generic real number,
$(u_l)_{l\in\mathbb{N}}$ the associated BKL or spike induced sequence of
epochs, and $(\u_s)_{s\in\mathbb{N}}$ the associated BKL or spike induced
sequence of eras. Consider the truncated sequence $(\u_1, \dotsc, \u_n)$ of
eras/$\era$s. Then the average $\langle\mathcal{W}_{\mathrm{K}}^2\rangle_n$
of $\mathcal{W}_{\mathrm{K}}^2$ over all epochs of the first $n$ BKL and
spike-induced eras satisfies
\begin{equation}\label{WBKL}
\langle\mathcal{W}_{\mathrm{K}}^2\rangle_n = 19.580317157(8) \:\frac{1}{\log_2 n} \:  \frac{1 + o(1)}{ 1 + \chi(n)}\:,
\end{equation}
where $\chi$ depends on if we consider an era or and $\era$, but in both
cases it is a function satisfying $\liminf_{n\rightarrow \infty} \chi(n) = 0$
and $\limsup_{n\rightarrow \infty} \chi(n) = \infty$.
\end{theorem}
\begin{remark}
Although the results are similarly the same for BKL and spike-induced
sequences of Kasner eras, note that the concepts of eras differ for the two
cases.
\end{remark}
\begin{remark}
The function $\chi$ does not converge to zero exactly, but `almost'.
Specifically, $\chi(n)$ converges to zero as $\mathcal{N} \ni n\rightarrow
\infty$, where $\mathcal{N}$ is a subset of $\mathbb{N}$ minus a subset of
logarithmic density zero. The average
$\langle\mathcal{W}_{\mathrm{K}}^2\rangle_n$ (slowly) converges to zero as
$n\rightarrow \infty$. Intuitively, this is due to the dominance of small
curvature phases over large curvature phases.
\end{remark}

\begin{proof}
We first prove Theorem~\ref{theoremW} for the BKL case and then for the
recurring spike case. In the BKL case we have:
\begin{lemma}\label{lemmaWBKL}
Let $\u_1 = u_1 = [k_1;k_2,k_3, \ldots]$ be a generic real number,
$(u_l)_{l\in\mathbb{N}}$ the associated BKL sequence of epochs, and
$(\u_s)_{s\in\mathbb{N}}$ the associated BKL sequence of eras. Consider the
truncated sequence $(\u_1, \dotsc, \u_n)$ of eras. Then the average
$\langle\mathcal{W}_{\mathrm{K}}^2\rangle_n$ of $\mathcal{W}_{\mathrm{K}}^2$
over all epochs of the first $n$ eras satisfies
\begin{equation}\label{WBKLn}
\langle\mathcal{W}_{\mathrm{K}}^2\rangle_n = 19.580317157(8) \:\frac{1}{\log_2 n} \:  \frac{1 + o(1)}{ 1 + \chi(n)}\:,
\end{equation}
where $\chi$ is a function satisfying  $\liminf_{n\rightarrow \infty} \chi(n)
= 0$ and $\limsup_{n\rightarrow \infty} \chi(n) = \infty$.
\end{lemma}

Let $n\gg 1$. The average of $\mathcal{W}_{\mathrm{K}}^2$ over all epochs ($=
k_1 + \dotsb + k_n$) of the first $n$ eras is
\[
\langle\mathcal{W}_{\mathrm{K}}^2\rangle_n =
\big(\sum_{i=1}^n k_i \big)^{-1} \sum_{l=1}^{k_1+\dotsb+k_n} \mathcal{W}_{\mathrm{K}}^2(u_l) =
\big(\sum_{i=1}^n k_i \big)^{-1} \sum_{s=1}^n \sum_{j=1}^{k_s} \mathcal{W}_{\mathrm{K}}^2\big(\{\u_s\} +j\big)\:,
\]
where $\{\u_s\}=[0;k_{s+1},\dotsc] = \u_{s+1}^{-1}$ is the fractional part of
$\u_s$; recall that the integer part $[\u_s]$ is equal to $k_s$. Since $n\gg
1$ we may consider $\{u_s\}=:\varkappa$ as a random variable on $(0,1)$; its
probability distribution is $w(\varkappa) = (\log 2)^{-1}
(1+\varkappa)^{-1}$, see~\cite{heietal09} with $\varkappa = u^{-1}$ and
$w(\varkappa) d\varkappa = -p(u) d u$. We are thus able to replace
$\sum_{j=1}^{k_s} \mathcal{W}_{\mathrm{K}}^2\big(\{\u_s\} +j\big)$ by its
weighted average $\overline{\mathcal{W}_{\mathrm{K}}^2}(k_s)$ which is given
by
\[
\overline{\mathcal{W}_{\mathrm{K}}^2}(m) :=
\int_0^1 w(\varkappa) \sum_{j=1}^{m} \mathcal{W}_{\mathrm{K}}^2(j + \varkappa) d\varkappa =
\sum_{j=1}^{m} \int_0^1 w(\varkappa) \mathcal{W}_{\mathrm{K}}^2(j + \varkappa) d\varkappa  \:.
\]
Accordingly we obtain
\begin{equation}\label{BKLWn1}
\langle\mathcal{W}_{\mathrm{K}}^2\rangle_n = \big(\sum_{i=1}^n k_i \big)^{-1} \sum_{s=1}^n \overline{\mathcal{W}_{\mathrm{K}}^2}(k_s)
= \frac{\sum_m \kh_n(m) \overline{\mathcal{W}_{\mathrm{K}}^2}(m)}{\frac{1}{n}\,\sum_{i=1}^n k_i} \:,
\end{equation}
where $\kh_n$ denotes the probability distribution associated with
$(k_1,\dotsc,k_n)$, which converges to $\kh$ as $n\rightarrow \infty$.

To compute the numerator of~\eqref{BKLWn1} we first note that
\begin{align}\label{BKLWaux}
\nonumber
\int_0^1 w(\varkappa) \mathcal{W}_{\mathrm{K}}^2(j + \varkappa) d\varkappa
& = \int_0^1 w(\varkappa) \mathcal{W}_{\mathrm{K}}^2(j + \bar{\varkappa}) d\varkappa +
\int_0^1 w(\varkappa) \big[ \mathcal{W}_{\mathrm{K}}^2(j+\varkappa) - \mathcal{W}_{\mathrm{K}}^2(j + \bar{\varkappa}) \big] d\varkappa \\
\nonumber
& =  \mathcal{W}_{\mathrm{K}}^2(j + \bar{\varkappa})
- 162 j^{-3} \int_0^1 w(\varkappa) (\varkappa - \bar{\varkappa}) d\varkappa + O(j^{-4}) \\
& = \int_{j+\bar{\varkappa}-1/2}^{j+\bar{\varkappa}+1/2} \mathcal{W}_{\mathrm{K}}^2(x) d x + O(j^{-4})
\end{align}
as $j\rightarrow \infty$, when we use $\bar{\varkappa} = \int_0^1 \varkappa
w(\varkappa) d\varkappa = (\log 2)^{-1} -1$. Accordingly, for $m > M\gg 1$ we
find
\begin{align*}
\overline{\mathcal{W}_{\mathrm{K}}^2}(m) - \overline{\mathcal{W}_{\mathrm{K}}^2}(M) & =
\sum_{j=M+1}^{m} \int_0^1 w(\varkappa) \mathcal{W}_{\mathrm{K}}^2(j + \varkappa) d\varkappa
= \int_{M+\bar{\varkappa}+1/2}^{m+\bar{\varkappa}+1/2}  \mathcal{W}_{\mathrm{K}}^2(x) d x + O(M^{-3}-m^{-3}) \\
& = 81 \big( M^{-1} -m^{-1} \big) - \frac{81}{\log 2}\: \big( M^{-2} - m^{-2} \big)  + O(M^{-3}-m^{-3})\:.
\end{align*}
Since $\sum_m \kh_n(m) \overline{\mathcal{W}_{\mathrm{K}}^2}(m)$ converges to
$\sum_m \kh(m) \overline{\mathcal{W}_{\mathrm{K}}^2}(m)$ as $n\rightarrow
\infty$, we consider the latter series. Choosing $\mathbb{N}\ni M \gg 1$ we
obtain
\begin{align*}
\sum_{m=1}^\infty \kh(m) \overline{\mathcal{W}_{\mathrm{K}}^2}(m) & =
\sum_{m=1}^M \kh(m) \overline{\mathcal{W}_{\mathrm{K}}^2}(m) +  \sum_{m= M+1}^\infty \kh(m)\:
\Big( \overline{\mathcal{W}_{\mathrm{K}}^2}(M) + \frac{81}{M}  -\frac{81}{m} + O\big(\frac{1}{M^2}\big) \Big)  \\
& = \sum_{m=1}^M \kh(m) \overline{\mathcal{W}_{\mathrm{K}}^2}(m) -
\big( \overline{\mathcal{W}_{\mathrm{K}}^2}(M) + \frac{81}{M} \big) \: \log_2 \frac{M+1}{M+2}
-\frac{81}{2 \log 2\: M^2} + O\big(\frac{1}{M^3}\big) \:.
\end{align*}
A numerical calculation then yields the result $19.580317157(8)$.

It remains to analyze the denominator of~\eqref{BKLWn1}, which is the mean
length of the first $n$ eras, or, in other words, the arithmetic mean of the
first $n$ partial quotients of $[k_1;k_2,\dotsc]$. From the theory of
continued fractions we have the result that
\begin{equation}\label{arimean}
\sum_{i=1}^n k_i = n \:\log_2 n \;\, \big( 1 + \chi(n) \big)\:,
\end{equation}
where $\chi$ is as described in the lemma and the remark; see,
e.g.,~\cite{kin64,Hardy/Wright:1979}. By combining~\eqref{arimean} with the
result of the previous calculation the lemma is proved.

Let us now turn to the recurring spike case. By considering the union of all
Kasner sequences and noting that the spike Kasner map is just the square of
the BKL map it is perhaps intuitive that we will obtain a similar result as
in the BKL case, but let us give a formal proof nonetheless.

%
\begin{lemma}\label{lemmaWnonBKL}
Let $\u_1 = u_1 = [k_1;k_2,k_3, \ldots]$ be a generic real number,
$(u_l)_{l\in\mathbb{N}}$ the associated spike induced sequence of epochs, and
$(\u_s)_{s\in\mathbb{N}}$ the associated spike induced sequence of eras.
Consider the truncated sequence $(\u_1, \dotsc, \u_n)$ of $\era$s. Then the
average $\langle\mathcal{W}_{\mathrm{K}}^2\rangle_n$ of
$\mathcal{W}_{\mathrm{K}}^2$ over all epochs of the first $n$ $\era$s
satisfies
\begin{equation}\label{nWBKL}
\langle\mathcal{W}_{\mathrm{K}}^2\rangle_n = 19.580317157(8) \:\frac{1}{\log_2 n} \:  \frac{1 + o(1)}{ 1 + \tilde{\chi}(n)}\:,
\end{equation}
where $\tilde{\chi}$ is a function with the same properties as $\chi$ of
Lemma~\ref{lemmaWBKL}.
\end{lemma}
\begin{remark}
The statement of this lemma is formally identical to that of
Lemma~\ref{lemmaWBKL}, which motivates the formulation of
Theorem~\ref{theoremW}. Note, however, that the concept of an era of a spike
induced sequence of Kasner epochs differs from that of a BKL sequence.
\end{remark}

For the proof of the above lemma we need~\eqref{kkappa}--\eqref{uspike}; in
particular we recall that $\u_s = [k_s^\prime, k_{s+\z_s + 1}, \dotsc]$. Let
us denote the length of $\era$ $s$ by $l_s$, where $l_s = l(k_s')$ with
\begin{equation}\label{landk'}
l(k') = \begin{cases} \frac{k^\prime}{2} & k' \;\,\mathrm{even} \:,\\
\frac{k'+1}{2} & k' \;\, \mathrm{odd} \:.
\end{cases}
\end{equation}
The average of $\mathcal{W}_{\mathrm{K}}^2$ over all epochs ($= l_1 + \dotsb
+ l_n$) of the first $n$ $\era$s is
\[
\langle\mathcal{W}_{\mathrm{K}}^2\rangle_n = \big(\sum_{i=1}^n l_i \big)^{-1} \sum_{l=1}^{l_1+\dotsb+l_n} \mathcal{W}_{\mathrm{K}}^2(u_l) =
\big(\sum_{i=1}^n l_i \big)^{-1} \sum_{s=1}^n \sum_{j=1}^{l_s} \mathcal{W}_{\mathrm{K}}^2\big(2 j - (k_s^\prime\,\mZ) + \{\u_s\}\big)\:.
\]
In complete analogy with the proof of Lemma~\ref{lemmaWBKL} we are able to
replace the inner sum by by the weighted average
$\overline{\mathcal{W}_{\mathrm{K}}^2}(k_s')$, where
\[
\overline{\mathcal{W}_{\mathrm{K}}^2}(k') :=
\int_0^1 w(\varkappa) \sum_{j=1}^{l(k')}  \mathcal{W}_{\mathrm{K}}^2\big(2 j - (k'\,\mZ) +\varkappa\big) d\varkappa =
\sum_{j=1}^{l(k')} \int_0^1 w(\varkappa)  \mathcal{W}_{\mathrm{K}}^2\big(2 j - (k'\,\mZ) +\varkappa\big) d\varkappa  \:.
\]
Accordingly,
\begin{equation}\label{nonBKLWn1}
\langle\mathcal{W}_{\mathrm{K}}^2\rangle_n = \big(\sum_{i=1}^n l_i \big)^{-1} \sum_{s=1}^n\overline{\mathcal{W}_{\mathrm{K}}^2}(k_s')
= \frac{\sum_{k'} \khp_n(k') \overline{\mathcal{W}_{\mathrm{K}}^2}(k')}{\frac{1}{n}\,\sum_{i=1}^n l_i}\:,
\end{equation}
where $\khp_n$ denotes the probability distribution associated with
$(k_1,\dotsc,k_n)$, which converges to $\khp$ as $n\rightarrow \infty$.

To compute the numerator of~\eqref{nonBKLWn1} we use~\eqref{BKLWaux}, where
$j$ is replaced by $2 j - k'\,\mZ$. In addition we make use of the relation
\[
\int_{2 j + a -1/2}^{2 j + a +1/2} \mathcal{W}_{\mathrm{K}}^2(x) d x = \frac{1}{4} \int_{2 j +
a -3/2}^{2 j + a -1/2} \mathcal{W}_{\mathrm{K}}^2 d x
+ \frac{1}{2} \int_{2 j + a -1/2}^{2 j + a + 1/2} \mathcal{W}_{\mathrm{K}}^2 d x
+ \frac{1}{4} \int_{2 j + a +1/2}^{2 j + a + 3/2} \mathcal{W}_{\mathrm{K}}^2 d x + O(j^{-4})\:,
\]
which holds for all $a$ (e.g., $\bar{\varkappa} - k' \:\mZ$). Let $k' > K'
\gg 1$ with $k' \:\mZ = K' \: \mZ$. We find
\begin{align*}
& \overline{\mathcal{W}_{\mathrm{K}}^2}(k') - \overline{\mathcal{W}_{\mathrm{K}}^2}(K') =
\sum_{j=l(K')+1}^{l(k')} \int_0^1 w(\varkappa)  \mathcal{W}_{\mathrm{K}}^2\big(2 j - (k'\,\mZ) + \varkappa\big) d\varkappa  = \\
& \quad\qquad= \frac{1}{4}\int_{K' + \bar{\varkappa} + 1/2}^{K'+\bar{\varkappa} + 3/2}  \mathcal{W}_{\mathrm{K}}^2(x) d x
+ \frac{1}{2} \int_{K'+\bar{\varkappa} + 3/2}^{k'+ \bar{\varkappa} + 1/2}  \mathcal{W}_{\mathrm{K}}^2(x) d x
+ \frac{1}{4} \int_{k' +\bar{\varkappa} + 1/2}^{k'+ \bar{\varkappa} + 3/2}   \mathcal{W}_{\mathrm{K}}^2(x) d x +O\big((K')^{-3}\big) \\[1ex]
& \quad\qquad=  \frac{81}{2} \big( (K')^{-1} -(k')^{-1} \big) - \frac{81}{4 \log 2}\, ( 2 + \log 2) \:
\big( (K')^{-2} - (k')^{-2} \big)  + O\big((K')^{-3}\big)
\end{align*}
Since $\sum_{k'} \khp_n(k') \overline{\mathcal{W}_{\mathrm{K}}^2}(k')$
converges to $\sum_{k'} \khp(k') \overline{\mathcal{W}_{\mathrm{K}}^2}(k')$
as $n\rightarrow \infty$, we consider the latter series. Choosing
$\mathbb{N}\ni K' \gg 1$ we obtain
\begin{align*}
\sum_{k'=K'+1}^\infty \khp(k') \overline{\mathcal{W}_{\mathrm{K}}^2}(k') & =
\sum_{\substack{k'= K'+1 \\k'\; \mathrm{odd}}}^\infty \khp(k')\:
\Big( \overline{\mathcal{W}_{\mathrm{K}}^2}(K'-1) + \frac{81}{2} \,\big(\frac{1}{K'}
- \frac{1}{k'}\big) + O\big(\frac{1}{K^{\prime\:2}}\big)\, \Big) \\[1ex]
 & \qquad +  \sum_{\substack{k'= K'+2 \\k'\; \mathrm{even}}}^\infty \khp(k')\:
\Big(\overline{\mathcal{W}_{\mathrm{K}}^2}(K')  + \frac{81}{2} \,\big(\frac{1}{K'}
- \frac{1}{k'}\big) + O\big(\frac{1}{K^{\prime\:2}}\big)\, \Big) \:,
\end{align*}
and thus
\begin{align*}
\sum_{k'=1}^\infty \khp(k') \overline{\mathcal{W}_{\mathrm{K}}^2}(k') & =
\sum_{k'=1}^{K'} \khp(k') \overline{\mathcal{W}_{\mathrm{K}}^2}(k')
-\overline{\mathcal{W}_{\mathrm{K}}^2}(K'-1) \:\log_3\frac{K'+1}{K'+2}
-\overline{\mathcal{W}_{\mathrm{K}}^2}(K')\: \log_3\frac{K'+2}{K'+3}  \\
& \qquad\qquad\qquad-\frac{81}{2} \,\frac{1}{K'} \: \log_3 \frac{K'+1}{K'+3}
-\frac{81}{2 \log 3}\: \frac{1}{K^{\prime\:2}} + O\big(\frac{1}{K^{\prime\:3}}\big) \:.
\end{align*}
A numerical calculation then yields the result $12.35380467921(1)$.

It remains to analyze the denominator of~\eqref{BKLWn1}, which is the mean
length of the first $n$ $\era$s. Using~\eqref{landk'} we have
\[
\sum_{i=1}^n l_i = \sum_{i=1}^n l(k'_i) = \frac{1}{2} \sum_{i=1}^n  k_i^\prime + \frac{1}{2} \, n\:\log_3 2 \:,
\]
where $\log_3 2$ is the probability for a randomly chosen element of
$(k_1',k_2',\dotsc)$ to be odd. From the construction of the auxiliary
sequence $(\kappa_1',\kappa_2',\dotsc, \kappa_{\hat{n}}')$ and the sequence
$(k_1', k_2',\dotsc, k_n')$, see~\eqref{kkappa}--\eqref{uspike}, we recall
that $(k_1',k_2',\dotsc,k_n')$ is shorter than $(\kappa_1',\kappa_2',\dotsc,
\kappa_{\hat{n}}')$ by the fraction $\frac{1}{2}\,\kh(1) = 1 - (2 \log_3
2)^{-1}$; accordingly, we have $\hat{n} = 2 (\log_3 2) \,n$. Furthermore, we
obtain
\[
\sum_{i=1}^n l_i = \frac{1}{2} \sum_{i=1}^{\hat{n}}  \kappa_i^\prime + \frac{1}{2} \, n\:\log_3 2
= \frac{1}{2} \sum_{i=1}^{\hat{n}} \big( k_i -\frac{1}{2} \big) + \frac{1}{2} \, n\:\log_3 2 \:,
\]
because, on the average, the elements of $(\kappa_1',\dotsc,
\kappa_{\hat{n}}')$ are smaller than the elements of $(k_1,\dotsc,
k_{\hat{n}})$ by one half; we again refer to the proof of
Lemma~\ref{nonKhinchinlemma}. Therefore,
\begin{equation}\label{noneramean}
\begin{split}
\sum_{i=1}^n l_i
& = \frac{1}{2} \sum_{i=1}^{\hat{n}}  k_i - \frac{1}{4}\, \hat{n} + \frac{1}{2} \, n\:\log_3 2
= \frac{1}{2} \sum_{i=1}^{\hat{n}}  k_i
= \frac{1}{2} \:\hat{n} \:\log_2 \hat{n} \:\,\big( 1 + \chi(\hat{n}) \big) \\
& = n \, \log_3 2 \: \big( 1 - \log_2\log_2 3 + \log_2 n \big)  \:\,\big( 1 + \chi(n) \big) =
n \, \log_3 n \: \big( 1 + \tilde{\chi}(n) \big)
\end{split}
\end{equation}
Noting that $12.35380467921(1)\, \log 3\, (\log 2)^{-1} = 19.580317157(8)$,
the lemma is proved. This concludes the proof of Theorem~\ref{theoremW}.
\end{proof}

\begin{remark}
The results of Lemma~\ref{lemmaWBKL} and~\ref{lemmaWnonBKL} differ slightly
when expressed in terms of a different parameter; instead of $n$ we use the
number $N$ of epochs up to (and including) era/$\era$ $n$. Let us restrict
ourselves to the main point. In the context of Lemma~\ref{lemmaWBKL} we have
$N = \sum_{i=1}^n k_i = n \:\log_2 n$, cf.~\eqref{arimean}. Therefore,
$\log_2 n = (\log 2)^{-1} W(N \log 2)$, where $W$ is the Lambert $W$ function
(product logarithm). For $N \gg 1$ we thus have
\begin{align*}
\log_2 n & = (\log 2)^{-1} \big( \log( N \log 2) - \log\log(N\log2) + o(1) \big)  \\
& =  (\log 2)^{-1} \big( \log N -\log\log N + \log \log 2 + o(1) \big) \:.
\end{align*}
Analogously, in the context of Lemma~\ref{lemmaWnonBKL} we have $N =
\sum_{i=1}^n l_i =  n \, \log_3 n$, cf.~\eqref{noneramean}. Therefore,
$\log_2 n = (\log 2)^{-1} W(N \log 3)$ and
\[
\log_2 n =  (\log 2)^{-1} \big( \log N -\log\log N + \log \log 3 + o(1) \big) \:.
\]
\end{remark}
%

\section{Concluding remarks}\label{sec:conc}

The BKL and the BKL era maps have attracted a lot of attention in the
literature, perhaps due to the tantalizing chaotic features that they
indicate, see~\cite{khaetal85} and~\cite{waiell97}--\cite{damlec11b}, and
references therein. The present paper opens up similar exploration
possibilities for recurring spike formation --- we have only given a few
examples of what can be done. As in the BKL case, one can derive
probabilities for other quantities, but one can also use the explicit
solutions that describe spike transitions to attempt to extend the associated
maps in perhaps similar ways as have been done for the BKL case. Another
possibility is to focus on the state space picture in~\cite{heietal12,ugg13},
where~\cite{heietal12} give some examples of so-called \emph{concatenated
chains}, where axes permutations are not quoted out (for a discussion,
see~\cite{ugg13}). In particular, one can explore periodic spike chains,
which are the recurring spike analogues of the so-called heteroclinic cycles
in the BKL case, which recently have attracted
attention~\cite{lieetal10,lieetal12}.

\subsection*{Acknowledgments}
We thank the Erwin Schr\"{o}dinger Institute, Vienna, for hospitality, where
part of this work was carried out.

\end{document}